\documentclass[a4paper,11pt,reqno]{amsart}
\usepackage[utf8]{inputenc}
\usepackage{amsfonts}
\usepackage{amsmath}
\usepackage{amssymb}
\usepackage{amsthm}
\usepackage{verbatim}
\usepackage{color}
\usepackage{hyperref}
\usepackage{epstopdf}

\newtheorem{theorem}{Theorem}[section]

\newtheorem{lemma}[theorem]{Lemma}
\newtheorem{conjecture}{Conjecture}[section]

\numberwithin{equation}{section}

\title[Random matrix theory and moments of moments of $L$-functions]{Random matrix theory and moments of moments of $L$-functions}

\author{J. C. Andrade}
\address{Department of Mathematics, University of Exeter, Exeter, EX4 4QF, United Kingdom}
\email{j.c.andrade@exeter.ac.uk}

\author{C. G. Best}
\address{Department of Mathematics, University of Exeter, Exeter, EX4 4QF, United Kingdom}
\email{cgb212@exeter.ac.uk}

\date{\today}


\begin{document}

\begin{abstract}
    We give an analytic proof of the asymptotic behaviour of the moments of moments of the characteristic polynomials of random symplectic and orthogonal matrices. We therefore obtain alternate, integral expressions for the leading order coefficients previously found by Assiotis, Bailey and Keating. We also discuss the conjectures of Bailey and Keating for the corresponding moments of moments of $L$-functions with symplectic and orthogonal symmetry. Specifically, we show that these conjectures follow from the shifted moments conjecture of Conrey, Farmer, Keating, Rubinstein and Snaith.
\end{abstract}

\maketitle

\section{Introduction}

Let $G(N)\in\{U(N),Sp(2N),SO(2N)\}$, where $U(N)$ is the group of $N\times N$ unitary matrices, $Sp(2N)$ is the the group of $2N\times 2N$ unitary symplectic matrices and $SO(2N)$ is the group of $2N\times 2N$ orthogonal matrices with determinant $+1$. Also, let

\begin{equation}
    P_{G(N)} (\theta; A)=\det\left(I-Ae^{-i\theta}\right),
\end{equation}
be the characteristic polynomial of a matrix $A\in G(N)$ on the unit circle. Recently, the {\it moments of moments} of these characteristic polynomials have been the object of much study. The moments of moments consist of an average over the unit circle first and then an average through the group, hence the name. Specifically, they are defined as

\begin{equation}
    \text{MoM}_{G(N)} (k,\beta):=\int_{G(N)} \left(\frac{1}{2\pi} \int_0^{2\pi} |P_{G(N)} (\theta; A)|^{2\beta} d\theta\right)^k dA,
\end{equation}
where $dA$ is the Haar measure on $G(N)$.

One particular motivation for the study of the moments of moments is their link to the maximum value of the characteristic polynomials on the unit circle. For example, in the case of the unitary group $U(N)$, Fyodorov, Hiary and Keating \cite{FHK} and Fyodorov and Keating \cite{FK}, using heuristics involving the moments of moments, made conjectures for the maximum value of $|P_{U(N)} (\theta;A)|$ for $0\leq \theta<2\pi$. For an in depth discussion of the conjectures of \cite{FHK,FK} and work in their direction, see \cite{BK1}.

Concerning the moments of moments, one of the conjectures of \cite{FK} is that as $N\to\infty$,
\begin{align}
    \text{MoM}_{U(N)}(k,\beta)\sim\begin{cases} \Big(\frac{\mathcal{G}(1+\beta)^2}{\mathcal{G}(1+2\beta)\Gamma(1-\beta^2)}\Big)^k \Gamma(1-k\beta^2) N^{k\beta^2} &\text{ if }k<1/\beta^2, \\ c(k,\beta) N^{k^2\beta^2-k+1} &\text{ if }k>1/\beta^2,\end{cases}
\end{align}
where $\mathcal{G}(s)$ is the Barnes $\mathcal{G}$-function and $c(k,\beta)$ is some unspecified function of $k$ and $\beta$. At the transition point $k=1/\beta^2$, the moments of moments are conjectured to grow like $N \log N$. The above conjecture was proven in the case that $k,\beta\in\mathbb{N}$ by Bailey and Keating \cite{BK2} with the following result.

\begin{theorem}[Bailey-Keating \cite{BK2}] \label{BK thm}
For $k,\beta\in\mathbb{N}$,

\begin{equation} \label{BK result}
    \textup{MoM}_{U(N)} (k,\beta)=c(k,\beta) N^{k^2\beta^2-k+1}\big(1+O\big(\tfrac{1}{N}\big)\big),
\end{equation}
where $c(k,\beta)$ can be written explicitly in the form of an integral. Furthermore, $\textup{MoM}_{U(N)} (k,\beta)$ is a polynomial in $N$ of degree $k^2\beta^2-k+1$.
\end{theorem}
The proof of theorem \ref{BK thm} uses the fact that for $k\in\mathbb{N}$, one can change the order of integration to obtain

\begin{equation}
    \text{MoM}_{U(N)} (k,\beta)=\frac{1}{(2\pi)^k}\int_0^{2\pi}\dotsi\int_0^{2\pi} I_{k,\beta}(U(N),\theta_1,\dots,\theta_k)\, d\theta_1\dots d\theta_k,
\end{equation}
where

\begin{equation}
    I_{k,\beta} (U(N),\underline{\theta})=\int_{U(N)} \prod_{j=1}^k |P_{U(N)} (\theta_j;A)|^{2\beta} dA.
\end{equation}
The function $I_{k,\beta} (U(N),\underline{\theta})$ is an autocorrelation function of the characteristic polynomials and was computed by Conrey et al. \cite{CFKRS1}. Two equivalent expressions for $I_{k,\beta} (U(N),\underline{\theta})$ are given in \cite{CFKRS1}; one takes the form of a combinatorial sum and the other is as a multiple contour integral. The first of these was used in \cite{BK2} to prove that $\text{MoM}_{U(N)}(k,\beta)$ is a polynomial and then a complex analytic argument using the integral representation was used to determine the asymptotic behaviour.

Assiotis and Keating \cite{AK} gave an alternate proof of the asymptotic formula in theorem \ref{BK thm} using a combinatorial approach involving constrained Gelfand-Tsetlin patterns. They therefore obtained a different expression for the leading order coefficient $c(k,\beta)$ as the volume of a certain region. It is remarked in \cite{AK} that their expression for $c(k,\beta)$ appears to be very difficult to obtain from the expression obtained in \cite{BK2}.

The approach used in \cite{AK} was then extended by Assiotis, Bailey and Keating \cite{ABK} to the symplectic and special orthogonal groups to determine the asymptotic behaviour of the moments of moments for $k,\beta\in\mathbb{N}$. Their results are stated below.

\begin{theorem}[Assiotis et al. \cite{ABK}] \label{ABK Sp thm}
    Let $k,\beta\in\mathbb{N}$. Then, $\textup{MoM}_{Sp(2N)}(k,\beta)$ is a polynomial function in $N$. Moreover,
    
    \begin{equation} \label{ABK Sp result}
        \textup{MoM}_{Sp(2N)}(k,\beta)=\mathfrak{c}_{Sp}(k,\beta) N^{k\beta(2k\beta+1)-k} \left(1+O\left(\tfrac{1}{N}\right)\right),
    \end{equation}
    where the leading order term coefficient $\mathfrak{c}_{Sp}(k,\beta)$ is the volume of a certain convex region and is strictly positive.
\end{theorem}

\begin{theorem}[Assiotis et al. \cite{ABK}] \label{ABK SO thm}
    Let $k,\beta\in\mathbb{N}$. Then, $\textup{MoM}_{SO(2N)}(k,\beta)$ is a polynomial function in $N$. Moreover,
    
    \begin{equation}
        \textup{MoM}_{SO(2N)}(1,1)=2(N+1)
    \end{equation}
    otherwise,
    
    \begin{equation} \label{ABK SO result}
        \textup{MoM}_{SO(2N)}(k,\beta)=\mathfrak{c}_{SO}(k,\beta) N^{k\beta(2k\beta-1)-k} \left(1+O\left(\tfrac{1}{N}\right)\right),
    \end{equation}
    where the leading order term coefficient $\mathfrak{c}_{SO}(k,\beta)$ is given as a sum of volumes of certain convex regions and is strictly positive.
\end{theorem}

Our goal here is to apply the complex analytic method used in \cite{BK2} to give an alternate proof of the asymptotic formulae in theorems \ref{ABK Sp thm} and \ref{ABK SO thm}. In particular, we obtain integral expressions for the leading order coefficients. Our main result is the following.

\begin{theorem} \label{MoM Sp thm}
    For $k,\beta\in\mathbb{N}$,
    
    \begin{equation}
        \textup{MoM}_{Sp(2N)} (k,\beta)=\gamma_{Sp} (k,\beta) N^{k\beta(2k\beta+1)-k} \left(1+O\left(\tfrac{1}{N}\right)\right),
    \end{equation}
    where $\gamma_{Sp} (k,\beta)$ is given explicitly in the form of an integral, see \textup{(\ref{gamma_Sp def})}.
\end{theorem}

\begin{theorem} \label{MoM SO thm}
    For $k,\beta\in\mathbb{N}$ with $(k,\beta)\neq (1,1)$,
    
    \begin{equation}
        \textup{MoM}_{SO(2N)} (k,\beta)=\gamma_{SO} (k,\beta) N^{k\beta(2k\beta-1)-k} \left(1+O\left(\tfrac{1}{N}\right)\right),
    \end{equation}
    where $\gamma_{SO} (k,\beta)$ is given explicitly in the form of an integral, see \textup{(\ref{gamma_SO def})}.
\end{theorem}

\subsection{Moments of moments of $L$-functions}

Also considered in \cite{FHK, FK} were the moments of moments of the Riemann zeta function. Analogously to the moments of moments of the characteristic polynomials, these consist of an average first over a short piece of the critical line and then an average over these intervals. Specifically, the moments of moments of $\zeta(s)$ are defined for $T>0$ and $\text{Re}(\beta)>-1/2$ by \footnote{In \cite{FHK, FK}, the intervals being averaged over were of length $2\pi$ rather than 1. However, in \cite{BK3}, the intervals were taken to be of length 1 for convenience. The analysis of \cite{BK3} holds for any interval that is $O(1)$ as $T\to\infty$.}

\begin{equation}
    \text{MoM}_{\zeta_T} (k,\beta):=\frac{1}{T}\int_0^T \left(\int_t^{t+1}|\zeta(\tfrac{1}{2}+ih)|^{2\beta} dh\right)^k dt.
\end{equation}
These moments of moments are also linked to the extreme values taken by the Riemann zeta function and in \cite{FK}, a conjecture for the local maximum on short intervals was put forward. There has been significant progress on this conjecture, see for example, \cite{ABBRS, ABR, Harper, Najnudel}.

Bailey and Keating \cite{BK3}, using the philosophy that the Riemann zeta function on the critical line can be modelled by the characteristic polynomial of random unitary matrices and theorem \ref{BK thm}, made the following conjecture.

\begin{conjecture} [Bailey-Keating \cite{BK3}] \label{BK conj}
For $k,\beta\in\mathbb{N}$,

\begin{equation}
    \textup{MoM}_{\zeta_T} (k,\beta)=\alpha(k,\beta) c(k,\beta) \left(\log\tfrac{T}{2\pi}\right)^{k^2\beta^2-k+1} \left(1+O_{k,\beta} \left(\log^{-1} T\right)\right),
\end{equation}
where $c(k,\beta)$ is the same coefficient appearing in \textup{(\ref{BK result})}, and $\alpha(k,\beta)$ contains the arithmetic information in the form of an Euler product.
\end{conjecture}

It was then proven in \cite{BK3} that conjecture \ref{BK conj} follows from the conjecture of Conrey et al. \cite{CFKRS2} for the shifted moments of the zeta function. Explicitly, they prove that a function which, according to the conjecture of \cite{CFKRS2} approximates $\text{MoM}_{\zeta_T} (k,\beta)$ up to a power saving in $T$, does indeed behave asymptotically as conjecture \ref{BK conj} predicts $\text{MoM}_{\zeta_T} (k,\beta)$ does. The proof is similar to that in \cite{BK2} due to the similarity of the integral expressions for the shifted moments of the unitary characteristic polynomials and the Riemann zeta function.

Finally, Bailey and Keating \cite{BK3} also considered the moments of moments of families $L$-functions with symplectic or orthogonal symmetry. For each of the symmetry types, the moments of moments consist of an average over a short interval near the symmetry point and then an average through the family. Using theorems \ref{ABK Sp thm} and \ref{ABK SO thm}, Bailey and Keating made conjectures for the asymptotic growth of the moments of moments of these families in the same spirit as that of conjecture \ref{BK conj}. Following the proof of both theorems \ref{MoM Sp thm} and \ref{MoM SO thm} in sections 2 and 3, we will look at the examples of a symplectic and orthogonal family of $L$-functions considered in \cite{BK3} and show that the corresponding conjectures of Bailey and Keating also follow from the shifted moments conjecture of \cite{CFKRS2}.

\section{The symplectic group $Sp(2N)$}

In this section we will prove theorem \ref{MoM Sp thm}. The argument follows that used in \cite{BK2} and makes use of the complex analytic techniques deployed in \cite{KO} and \cite{KRRR}. The eigenvalues of matrices in $Sp(2N)$ lie on the unit circle and come in complex conjugate pairs $e^{i\phi_1}, e^{-i\phi_1}, e^{i\phi_2}, e^{-i\phi_2},\dots, e^{i\phi_N}, e^{-i\phi_N}$. Hence,

\begin{equation}\label{ref 1}
    \overline{P_{Sp(2N)} (\theta; A)}=P_{Sp(2N)} (-\theta; A).
\end{equation}

For $k,\beta\in\mathbb{N}$, we can change the order of integration by Fubini's theorem and use (\ref{ref 1}) to see that

\begin{equation}
    \text{MoM}_{Sp(2N)} (k,\beta)=\frac{1}{(2\pi)^k} \int_0^{2\pi}\cdots\int_0^{2\pi} I_{k,\beta} (Sp(2N),\theta_1,\dots,\theta_k)\, d\theta_1\cdots d\theta_k,
\end{equation}
where

\begin{equation}
    I_{k,\beta} (Sp(2N),\underline{\theta}):=\int_{Sp(2N)} \prod_{j=1}^k P_{Sp(2N)} (\theta_j; A)^{\beta} P_{Sp(2N)} (-\theta_j; A)^{\beta} dA.
\end{equation}
The autocorrelation function $I_{k,\beta} (Sp(2N),\underline{\theta})$ was also calculated by Conrey et al. \cite{CFKRS1}. In particular, it can written in the form of a multiple contour integral as 

\begin{align} \label{I(Sp) def}
    I_{k,\beta} (Sp(2N),\underline{\theta})=& \frac{(-1)^{k\beta} 2^{2k\beta}}{(2\pi i)^{2k\beta} (2k\beta)!} \oint\cdots\oint \prod_{1\leq m\leq n\leq 2k\beta} \left(1-e^{-z_m-z_n}\right)^{-1} \nonumber \\
    &\qquad \times\frac{\Delta(z_1^2,\dots,z_{2k\beta}^2)^2 \prod_{n=1}^{2k\beta} z_n}{\prod_{n=1}^{2k\beta} \prod_{m=1}^k (z_n-i\theta_m)^{2\beta} (z_n+i\theta_m)^{2\beta}} e^{N \sum_{n=1}^{2k\beta} z_n} dz_1\cdots dz_{2k\beta},
\end{align}
where $\Delta(z_1,\dots,z_n)=\prod_{i<j} (z_j-z_i)$ is the Vandermonde determinant and the contours encircle the poles at $\pm i\theta_m$ for $1\leq m\leq k$.

Each of the $2k\beta$ contours in (\ref{I(Sp) def}) can be deformed into $2k$ small circles around each of the poles $\pm i\theta_m$ with connecting straight lines whose contributions will cancel out. The multiple integral $I_{k,\beta} (Sp(2N),\underline{\theta})$ can therefore be written as a sum of $(2k)^{2k\beta}$ integrals. For $j\in\{\pm 1,\dots,\pm k\}$, let $C_j$ denote a small circular contour around $i\theta_j$ if $j>0$ and a small circular contour around $-i\theta_{-j}$ if $j<0$. Then we have that

\begin{equation} \label{ref 2}
    I_{k,\beta} (Sp(2N),\underline{\theta})=\frac{(-1)^{k\beta} 2^{2k\beta}}{(2\pi i)^{2k\beta} (2k\beta)!} \sum_{\varepsilon_j \in\{\pm 1,\dots,\pm k\}} J_{k,\beta} (\underline{\theta};\varepsilon_1,\dots,\varepsilon_{2k\beta}),
\end{equation}
where

\begin{align}
    J_{k,\beta} (\underline{\theta};\underline{\varepsilon})=& \int_{C_{\varepsilon_{2k\beta}}} \cdots\int_{C_{\varepsilon_1}} \prod_{1\leq m\leq n\leq 2k\beta} \left(1-e^{-z_m-z_n}\right)^{-1} \nonumber \\
    &\qquad \times\frac{\Delta(z_1^2,\dots,z_{2k\beta}^2)^2 \prod_{n=1}^{2k\beta} z_n}{\prod_{n=1}^{2k\beta} \prod_{m=1}^k (z_n-i\theta_m)^{2\beta} (z_n+i\theta_m)^{2\beta}} e^{N \sum_{n=1}^{2k\beta} z_n} dz_1\cdots dz_{2k\beta}.
\end{align}

Many of the summands in (\ref{ref 2}) are in fact zero as the following lemma demonstrates.

\begin{lemma} \label{zero summands}
For a choice of contours $\underline{\varepsilon}=(\varepsilon_1,\dots,\varepsilon_{2k\beta})$ in \textup{(\ref{ref 2})} and for $j\in\{1,\dots,k\}$, let $m_j$ and $n_j$ be the number of occurrences of $j$ and $-j$ respectively in $\underline{\varepsilon}$. Then, if $m_j+n_j>2\beta$ for some $j$, we have that $J_{k,\beta} (\underline{\theta};\underline{\varepsilon})$ is identically zero.
\end{lemma}

\begin{proof}
The proof is similar to that of lemma 3.2 in \cite{BK2} and lemma 4.11 in \cite{KRRR} but we include it for the sake of completeness. Suppose without loss of generality that $m_1+n_1=2\beta+1$ and that

\begin{equation}
    \underline{\varepsilon}=(\underbrace{1,\dots,1}_{m_1},\underbrace{-1,\dots,-1}_{2\beta+1-m_1},\underbrace{2,\dots,2}_{\beta-1},\underbrace{-2,\dots,-2}_{\beta},\underbrace{3,\dots,3}_{\beta},\underbrace{-3,\dots,-3}_{\beta},\dots,\underbrace{k,\dots,k}_{\beta},\underbrace{-k,\dots,-k}_{\beta}).
\end{equation}
All other cases can be proven in the same way. For simplicity, we assume that $m_1=0$. If this is not the case, then to $J_{k,\beta} (\underline{\theta};\underline{\varepsilon})$, we would make the change of variables $z_j\mapsto -z_j$ for $1\leq j\leq m_1$ and the same argument applies.

Making the change of variables $z_j\mapsto z_j-i\theta_1$, the integrand of $J_{k,\beta} (\underline{\theta};\underline{\varepsilon})$ is then

\begin{equation}
    \frac{G(z_1,\dots,z_{2\beta+1}) \Delta\left((z_1-i\theta_1)^2,\dots,(z_{2k\beta}-i\theta_1)^2 \right) dz_1\cdots dz_{2k\beta}}{\prod_{n=1}^{2\beta+1} z_n^{2\beta}},
\end{equation}
where

\begin{align}
    & G(z_1,\dots,z_{2\beta+1})= \nonumber \\
    & \frac{\prod_{1\leq m\leq n\leq 2k\beta} \left(1-e^{-z_m-z_n+2i\theta_1}\right)^{-1} \Delta\left((z_1-i\theta_1)^2,\dots,(z_{2k\beta}-i\theta_1)^2 \right) \prod_{n=1}^{2k\beta} (z_n-i\theta_1) e^{N \sum_{n=1}^{2k\beta} (z_n-i\theta_1)}}{\prod_{n=1}^{2k\beta} \prod_{m=2}^k \left(z_n-i(\theta_1+\theta_m)\right)^{2\beta} \left(z_n+i(\theta_m-\theta_1)\right)^{2\beta} \prod_{n=1}^{2k\beta} (z_n-2i\theta_1)^{2\beta} \prod_{n=2\beta+2}^{2k\beta} z_n^{2\beta}}
\end{align}
is analytic around the origin. The idea now is to show that the coefficient of $\prod_{n=1}^{2\beta+1} z_n^{-1}$ in the integrand of $J_{k,\beta} (\underline{\theta};\underline{\varepsilon})$ is zero and hence by the residue theorem, so is the integral. We have seen that $G(z_1,\dots,z_{2\beta+1})$ is analytic around zero and we can write the Vandermonde as

\begin{align}
    \Delta\left((z_1-i\theta_1)^2,\dots,(z_{2k\beta}-i\theta_1)^2 \right)=& \Delta\left((z_1^2-2i\theta_1 z_1),\dots,(z_{2k\beta}^2-2i\theta_1 z_{2k\beta})\right) \nonumber \\
    =& \sum_{\sigma\in S_{2k\beta}} \text{sgn}(\sigma) \prod_{n=1}^{2k\beta} (z_n^2-2i\theta_1 z_n)^{\sigma(n)-1}.
\end{align}
For each permutation $\sigma\in S_{2k\beta}$, we must have $\sigma(n)-1\geq 2\beta$ for at least one $n\in\{1,2,\dots,2\beta+1\}$. It follows that there are no terms in the expansion of the Vandermonde of the form $\prod_{n=1}^{2\beta+1} z_n^{a(n)}$ with $a(n)\leq 2\beta-1$. Thus, as $G(z_1,\dots,z_{2\beta+1})$ is analytic around zero, the coefficient of $\prod_{n=1}^{2\beta+1} z_n^{-1}$ in the integrand of $J_{k,\beta} (\underline{\theta};\underline{\varepsilon})$ is zero which completes the proof.
\end{proof}

Lemma \ref{zero summands} implies that the non-zero summands in (\ref{ref 2}) are given by those $\underline{\varepsilon}$ for which $m_j+n_j=2\beta$ for all $j$. This and the fact that the integrand of $J_{k,\beta} (\underline{\theta};\underline{\varepsilon})$ is a symmetric function of $z_1,\dots,z_{2k\beta}$, means we can rewrite (\ref{ref 2}) as

\begin{equation} \label{ref 6}
    I_{k,\beta} (Sp(2N),\underline{\theta})=\frac{(-1)^{k\beta} 2^{2k\beta}}{(2\pi i)^{2k\beta} (2k\beta)!} \sum_{l_1=0}^{2\beta} \cdots \sum_{l_k=0}^{2\beta} c_{\underline{l}} (k,\beta) J_{k,\beta} (\underline{\theta};l_1,\dots,l_k),
\end{equation}
where $J_{k,\beta} (\underline{\theta};\underline{l})$ is defined to be $J_{k,\beta} (\underline{\theta};\underline{\varepsilon})$ with $\underline{\varepsilon}$ given by

\begin{equation}
    \underline{\varepsilon}=(\underbrace{1,\dots,1}_{l_1},\underbrace{-1,\dots,-1}_{2\beta-l_1},\underbrace{2,\dots,2}_{l_2},\underbrace{-2,\dots,-2}_{2\beta-l_2},\dots,\underbrace{k,\dots,k}_{l_k},\underbrace{-k,\dots,-k}_{2\beta-l_k}),
\end{equation}
and

\begin{equation}\label{c(k,beta) def}
    c_{\underline{l}}(k,\beta)=\binom{2k\beta}{l_1}\binom{2k\beta-l_1}{2\beta-l_1}\binom{(2k-2)\beta}{l_2}\binom{(2k-2)\beta-l_2}{2\beta-l_2}\dotsi\binom{2\beta}{l_k}
\end{equation}
counts the number of ways in which $\underline{\varepsilon}$ can be permuted. The next lemma determines the asymptotic behaviour of $ I_{k,\beta} (Sp(2N),\underline{\theta})$.

\begin{lemma} \label{I(Sp) asymp}
As $N\to\infty$, we have

\begin{align}
    I_{k,\beta} (Sp(2N),\underline{\theta})\sim& \sum_{l_1,\dots,l_k=0}^{2\beta} \frac{(-1)^{k\beta+\sum_{m=1}^k l_m} c_{\underline{l}} (k,\beta)}{(2\pi i)^{2k\beta} (2k\beta)!} N^{|\mathcal{A}_{k,\beta;\underline{l}}|} e^{iN \sum_{n=1}^{2k\beta} \mu_n} \nonumber \\
    &\qquad \times\int_{C_0}\cdots\int_{C_0} \prod_{\substack{1\leq m\leq n\leq 2k\beta \\ \mu_m+\mu_n\neq0}}\left(1-e^{-\frac{(v_m+v_n)}{N}-i(\mu_m+\mu_n)}\right)^{-1} f(\underline{v};\underline{l}) \prod_{n=1}^{2k\beta} dv_n,
\end{align}
where $C_0$ denotes a small circular contour around the origin. The set $\mathcal{A}_{k,\beta;\underline{l}}$ and the function $f(\underline{v};\underline{l})$ are defined in the proof, see \textup{(\ref{set A def})} and \textup{(\ref{f(v) def})} respectively. Also, the $\mu_n$ are defined in terms of the $\theta_m$ in \textup{(\ref{mu def})}.
\end{lemma}

\begin{proof}

In view of (\ref{ref 6}),  we focus on $J_{k,\beta} (\underline{\theta};\underline{l})$. For a given $\underline{l}$, we make the change of variables

\begin{equation}
    z_n=\frac{v_n}{N}+i\mu_n,
\end{equation}
where

\begin{align} \label{mu def}
    \mu_n=\begin{cases} \theta_1, &\text{if } 1\leq n\leq l_1 \\
                        -\theta_1, &\text{if } l_1+1\leq n\leq 2\beta \\
                        \theta_2, &\text{if } 2\beta+1\leq n\leq 2\beta+l_2 \\
                        -\theta_2, &\text{if } 2\beta+l_2+1\leq n\leq 4\beta \\
                        \vdots &\vdots \\
                        \theta_k, &\text{if } (2k-2)\beta+1\leq n\leq (2k-2)\beta+l_k \\
                        -\theta_k, &\text{if } (2k-2)\beta+l_k+1\leq n\leq 2k\beta. \end{cases}
\end{align}
The contours of integration are then all small circles around the origin. Now, since the Laurent expansion of $(1-e^{-s})^{-1}$ about its pole at $s=0$ is

\begin{equation}
    (1-e^{-s})^{-1}=\frac{1}{s}+O(1),
\end{equation}
the integrand of $J_{k,\beta} (\underline{\theta};\underline{l})$ as $N\to\infty$ is then

\begin{align} \label{J integrand}
    & \prod_{1\leq m\leq n\leq 2k\beta} \left(1-e^{-\frac{(v_m+v_n)}{N}-i(\mu_m+\mu_n)}\right)^{-1} \Delta\left(\left(\tfrac{v_1}{N}+i\mu_1\right)^2,\dots,\left(\tfrac{v_{2k\beta}}{N}+i\mu_{2k\beta}\right)^2 \right)^2 \nonumber \\
    &\qquad \times\frac{\prod_{n=1}^{2k\beta}\big(\tfrac{v_n}{N}+i\mu_n\big)}{\prod_{n=1}^{2k\beta}\prod_{m=1}^k\big(\tfrac{v_n}{N}+i(\mu_n-\theta_m)\big)^{2\beta}\big(\tfrac{v_n}{N}+i(\mu_n+\theta_m)\big)^{2\beta}} e^{\sum_{n=1}^{2k\beta}v_n} e^{iN\sum_{n=1}^{2k\beta}\mu_n} \prod_{n=1}^{2k\beta}\frac{dv_n}{N} \nonumber \\
    =& \left(1+O\left(\tfrac{1}{N}\right)\right) N^{-2k\beta} \prod_{\substack{1\leq m\leq n\leq 2k\beta\\ \mu_m+\mu_n\neq 0}} \left(1-e^{-\frac{(v_m+v_n)}{N}-i(\mu_m+\mu_n)}\right)^{-1} e^{\sum_{n=1}^{2k\beta}v_n} e^{iN\sum_{n=1}^{2k\beta}\mu_n} \nonumber \\
    &\qquad \times\prod_{n=1}^{2k\beta}(i\mu_n) \frac{\prod_{\substack{1\leq m<n\leq2k\beta\\ \mu_m+\mu_n\neq0}}(i\mu_m+i\mu_n)^2 \prod_{\substack{1\leq m<n\leq2k\beta\\ \mu_m-\mu_n\neq 0}}(i\mu_m-i\mu_n)^2}{\prod_{\substack{1\leq m<n\leq2k\beta\\ \mu_m+\mu_n=0}}\Big(\frac{N}{v_n+v_m}\Big) \prod_{\substack{1\leq m<n\leq2k\beta\\ \mu_m-\mu_n=0}}\Big(\frac{N}{v_n-v_m}\Big)^2} \nonumber \\
    &\qquad \times\frac{\prod_{n=1}^{2k\beta}\big(\tfrac{v_n}{N}\big)^{-2\beta}}{\underset{\mu_n-\theta_m\neq0}{\prod_{n=1}^{2k\beta}\prod_{m=1}^k}(i\mu_n-i\theta_m)^{2\beta}\underset{\mu_n+\theta_m\neq0}{\prod_{n=1}^{2k\beta}\prod_{m=1}^k}(i\mu_n+i\theta_m)^{2\beta}} \prod_{n=1}^{2k\beta}dv_n \nonumber \\
    =& \left(1+O\left(\tfrac{1}{N}\right)\right) (-1)^{k\beta} N^{4k\beta^2-2k\beta} \prod_{\substack{1\leq m\leq n\leq 2k\beta \\ \mu_m+\mu_n\neq0}}\left(1-e^{-\frac{(v_m+v_n)}{N}-i(\mu_m+\mu_n)}\right)^{-1} e^{\sum_{n=1}^{2k\beta}v_n} e^{iN\sum_{n=1}^{2k\beta}\mu_n} \nonumber \\
    &\qquad \times\prod_{n=1}^{2k\beta}\mu_n \frac{\prod_{\substack{m<n\\ \mu_m^2=\mu_n^2}}(2i\mu_n)^2 \prod_{\substack{m<n\\ \mu_m^2\neq\mu_n^2}}(\mu_m^2-\mu_n^2)^2}{\prod_{\substack{1\leq m<n\leq2k\beta\\ \mu_m+\mu_n=0}}\Big(\frac{N}{v_n+v_m}\Big) \prod_{\substack{1\leq m<n\leq2k\beta\\ \mu_m-\mu_n=0}}\Big(\frac{N}{v_n-v_m}\Big)^2} \nonumber \\
    &\qquad \times\frac{1}{\underset{\mu_n^2=\theta_m^2}{\prod_{n=1}^{2k\beta}\prod_{m=1}^k}(2i\theta_m)^{2\beta} \underset{\mu_n^2\neq \theta_m^2}{\prod_{n=1}^{2k\beta}\prod_{m=1}^k}(\mu_n^2-\theta_m^2)^{2\beta}} \prod_{n=1}^{2k\beta}\frac{dv_n}{v_n^{2\beta}} \nonumber \\
    =& \left(1+O\left(\tfrac{1}{N}\right)\right) \frac{(-1)^{\sum_{m=1}^k l_m}}{2^{2k\beta}} N^{4k\beta^2-2k\beta} \prod_{\substack{1\leq m\leq n\leq 2k\beta \\ \mu_m+\mu_n\neq0}}\left(1-e^{-\frac{(v_m+v_n)}{N}-i(\mu_m+\mu_n)}\right)^{-1} e^{\sum_{n=1}^{2k\beta}v_n} e^{iN\sum_{n=1}^{2k\beta}\mu_n} \nonumber \\
    &\qquad \times\prod_{\substack{1\leq m<n\leq 2k\beta\\ \mu_m+\mu_n=0}} \left(\frac{v_n+v_m}{N}\right) \prod_{\substack{1\leq m<n\leq 2k\beta\\ \mu_m-\mu_n=0}} \left(\frac{v_n-v_m}{N}\right)^2 \prod_{n=1}^{2k\beta} \frac{dv_n}{v_n^{2\beta}}.
\end{align}

The power of $N$ coming from the products in the second line of (\ref{J integrand}) is determined by the size of the following sets:

\begin{equation}\label{set A def}
    \mathcal{A}_{k,\beta;\underline{l}}:=\{(m,n):1\leq m<n\leq 2k\beta, \mu_m+\mu_n=0\},
\end{equation}

\begin{equation}
    \mathcal{B}_{k,\beta;\underline{l}}:=\{(m,n):1\leq m<m\leq 2k\beta, \mu_m-\mu_n=0\}.
\end{equation}
Using the definition of $\mu_1,\dots,\mu_{2k\beta}$, we have that

\begin{equation}
    |\mathcal{A}_{k,\beta;\underline{l}}|=\sum_{m=1}^k l_m(2\beta-l_m),
\end{equation}
and

\begin{equation}
    |\mathcal{B}_{k,\beta;\underline{l}}|=k\beta(2\beta-1)+\sum_{m=1}^k l_m(l_m-2\beta).
\end{equation}
In particular,

\begin{equation}
    |\mathcal{A}_{k,\beta;\underline{l}}|+|\mathcal{B}_{k,\beta;\underline{l}}|=\#\{(m,n):1\leq m<n\leq 2k\beta, \mu_m^2=\mu_n^2\}=k\beta(2\beta-1),
\end{equation}
and the power of $N$ in the second line of (\ref{J integrand}) is

\begin{equation}
    -|\mathcal{A}_{k,\beta;\underline{l}}|-2|\mathcal{B}_{k,\beta;\underline{l}}|=|\mathcal{A}_{k,\beta;\underline{l}}|-2k\beta(2\beta-1).
\end{equation}

Therefore the integrand of $J_{k,\beta} (\underline{\theta};\underline{l})$ as $N\to\infty$ is

\begin{equation}
    \left(1+O\left(\tfrac{1}{N}\right)\right) \frac{(-1)^{\sum_{m=1}^k l_m}}{2^{2k\beta}} N^{|\mathcal{A}_{k,\beta;\underline{l}}|} \prod_{\substack{1\leq m\leq n\leq 2k\beta \\ \mu_m+\mu_n\neq0}}\left(1-e^{-\frac{(v_m+v_n)}{N}-i(\mu_m+\mu_n)}\right)^{-1} e^{iN \sum_{n=1}^{2k\beta} \mu_n} f(\underline{v};\underline{l}) \prod_{n=1}^{2k\beta} dv_n,
\end{equation}
where

\begin{equation} \label{f(v) def}
    f(\underline{v};\underline{l}):=\frac{\prod_{\substack{1\leq m<n\leq2k\beta\\ \mu_m+\mu_n=0}}(v_n+v_m) \prod_{\substack{1\leq m<n\leq2k\beta\\ \mu_m-\mu_n=0}}(v_n-v_m)^2}{\prod_{n=1}^{2k\beta} v_n^{2\beta}} e^{\sum_{n=1}^{2k\beta} v_n}
\end{equation}
does not depend on $\theta_1,\dots,\theta_k$. Hence, as $N\to\infty$,

\begin{align}
    I_{k,\beta} (Sp(2N),\underline{\theta})\sim& \sum_{l_1,\dots,l_k=0}^{2\beta} \frac{(-1)^{k\beta+\sum_{m=1}^k l_m} c_{\underline{l}} (k,\beta)}{(2\pi i)^{2k\beta} (2k\beta)!} N^{|\mathcal{A}_{k,\beta;\underline{l}}|} e^{iN \sum_{n=1}^{2k\beta} \mu_n} \nonumber \\
    &\qquad \times\int_{C_0}\cdots\int_{C_0} \prod_{\substack{1\leq m\leq n\leq 2k\beta \\ \mu_m+\mu_n\neq0}}\left(1-e^{-\frac{(v_m+v_n)}{N}-i(\mu_m+\mu_n)}\right)^{-1} f(\underline{v};\underline{l}) \prod_{n=1}^{2k\beta} dv_n,
\end{align}
where $C_0$ denotes a small circular contour around the origin.

\end{proof}

We can now obtain an asymptotic formula for $\textup{MoM}_{Sp(2N)} (k,\beta)$.

\begin{lemma} \label{MoM_Sp asymp lemma}
As $N\to\infty$, we have

\begin{equation}
    \textup{MoM}_{Sp(2N)} (k,\beta)\sim \gamma_{Sp} (k,\beta) N^{k\beta(2k\beta+1)-k},
\end{equation}
where $\gamma_{Sp} (k,\beta)$ is given in the form of an integral and is defined on the proof, see \textup{(\ref{gamma_Sp def})}.
\end{lemma}

\begin{proof}

By lemma \ref{I(Sp) asymp}, we have

\begin{align}
    \text{MoM}_{Sp(2N)} (k,\beta)=& \frac{1}{(2\pi)^k} \int_0^{2\pi}\cdots\int_0^{2\pi} I_{k,\beta} (Sp(2N),\theta_1,\dots,\theta_k) d\theta_1\cdots d\theta_k \nonumber \\
    \sim& \frac{1}{(2\pi)^k}\sum_{l_1,\dots,l_k=0}^{2\beta} \frac{(-1)^{k\beta+\sum_{m=1}^k l_m}c_{\underline{l}}(k,\beta)}{(2\pi i)^{2k\beta}(2k\beta)!} N^{|\mathcal{A}_{k,\beta;\underline{l}}|} \int_0^{2\pi}\dotsi\int_0^{2\pi} e^{iN\sum_{n=1}^{2k\beta}\mu_n} \nonumber \\
    &\times \int_{C_0}\dotsi\int_{C_0} \prod_{\substack{1\leq m\leq n\leq 2k\beta \\ \mu_m+\mu_n\neq0}}\Big(1-e^{-\frac{(v_m+v_n)}{N}-i(\mu_m+\mu_n)}\Big)^{-1} f(\underline{v};\underline{l}) \prod_{n=1}^{2k\beta}dv_n \prod_{m=1}^k d\theta_m.
\end{align}
Changing the order of integration, we have that

\begin{align} \label{MoM_Sp asymp}
    & \text{MoM}_{Sp(2N)} (k,\beta) \nonumber \\
    &\qquad \sim\sum_{l_1,\dots,l_k=0}^{2\beta} \frac{(-1)^{k\beta+\sum_{m=1}^k l_m}c_{\underline{l}}(k,\beta)}{(2\pi)^k (2\pi i)^{2k\beta}(2k\beta)!} N^{|\mathcal{A}_{k,\beta;\underline{l}}|} \int_{C_0}\dotsi\int_{C_0} f(\underline{v};\underline{l}) \nonumber \\
    &\qquad\qquad \times\int_0^{2\pi}\dotsi\int_0^{2\pi} \prod_{\substack{1\leq m\leq n\leq 2k\beta \\ \mu_m+\mu_n\neq0}}\Big(1-e^{-\frac{(v_m+v_n)}{N}-i(\mu_m+\mu_n)}\Big)^{-1} e^{iN\sum_{n=1}^{2k\beta}\mu_n} \prod_{m=1}^k d\theta_m \prod_{n=1}^{2k\beta}dv_n,
\end{align}
and we now seek to determine the $N$ dependence of the inner integrals over $\theta_1,\dots,\theta_k$. The first step is to write the integrand explicitly in terms of $\theta_1,\dots,\theta_k$ using the definition of $\mu_1,\dots,\mu_{2k\beta}$. The exponential term is

\begin{equation}
    \exp\left(iN \sum_{n=1}^{2k\beta} \mu_n\right)=\exp\left(2iN \sum_{m=1}^k (l_m-\beta)\theta_m\right).
\end{equation}

For the product of $(1-e^{-z_m-z_n})^{-1}$ terms, we define the set

\begin{equation}
    \mathcal{T}_{k,\beta;\underline{l}}:=\{(m,n):1\leq m\leq n\leq 2k\beta, \mu_m+\mu_n\neq0\}=\{(m,n):1\leq m\leq n\leq 2k\beta\}\setminus\mathcal{A}_{k,\beta;\underline{l}},
\end{equation}
and the following disjoint subsets of $\mathcal{T}_{k,\beta;\underline{l}}$ for $1\leq\sigma\leq\tau\leq k$:

\begin{equation}\label{U+ def}
    \mathcal{U}_{\sigma,\tau;\underline{l}}^{+}:=\{(m,n)\in\mathcal{T}_{k,\beta;\underline{l}}: \mu_m+\mu_n=\theta_{\sigma}+\theta_{\tau}\},
\end{equation}

\begin{equation}\label{U- def}
    \mathcal{U}_{\sigma,\tau;\underline{l}}^{-}:=\{(m,n)\in\mathcal{T}_{k,\beta;\underline{l}}: \mu_m+\mu_n=-(\theta_{\sigma}+\theta_{\tau})\},
\end{equation}
and

\begin{equation}\label{V+ def}
    \mathcal{V}_{\sigma,\tau;\underline{l}}^{+}:=\{(m,n)\in\mathcal{T}_{k,\beta;\underline{l}}: \mu_m+\mu_n=\theta_{\sigma}-\theta_{\tau}\},
\end{equation}

\begin{equation}\label{V- def}
    \mathcal{V}_{\sigma,\tau;\underline{l}}^{-}:=\{(m,n)\in\mathcal{T}_{k,\beta;\underline{l}}: \mu_m+\mu_n=-(\theta_{\sigma}-\theta_{\tau})\}.
\end{equation}
Note that $\mathcal{V}_{\sigma,\tau;\underline{l}}^{+}=\mathcal{V}_{\sigma,\tau;\underline{l}}^{-}=\emptyset$ for $\sigma=\tau$. The product of $(1-e^{-z_m-z_n})^{-1}$ terms can then be written as

\begin{align} \label{(1-e^s) prod}
    &\prod_{\substack{1\leq m\leq n\leq 2k\beta \\ \mu_m+\mu_n\neq0}}\Big(1-e^{-\frac{(v_m+v_n)}{N}-i(\mu_m+\mu_n)}\Big)^{-1} \nonumber \\
    &\qquad =\prod_{1\leq\sigma\leq\tau\leq k} \prod_{(m,n)\in\mathcal{U}_{\sigma,\tau;\underline{l}}^{+}} \Big(1-e^{-\frac{(v_m+v_n)}{N}-i(\theta_{\sigma}+\theta_{\tau})}\Big)^{-1} \prod_{(m,n)\in\mathcal{U}_{\sigma,\tau;\underline{l}}^{-}} \Big(1-e^{-\frac{(v_m+v_n)}{N}+i(\theta_{\sigma}+\theta_{\tau})}\Big)^{-1} \nonumber \\
    &\qquad\qquad \times\prod_{(m,n)\in\mathcal{V}_{\sigma,\tau;\underline{l}}^{+}} \Big(1-e^{-\frac{(v_m+v_n)}{N}-i(\theta_{\sigma}-\theta_{\tau})}\Big)^{-1} \prod_{(m,n)\in\mathcal{V}_{\sigma,\tau;\underline{l}}^{-}} \Big(1-e^{-\frac{(v_m+v_n)}{N}+i(\theta_{\sigma}-\theta_{\tau})}\Big)^{-1}.
\end{align}
Now, we make the change of variables $t_m=N \theta_m$. As $N\to\infty$, by the Laurent expansion of $(1-e^{-s})^{-1}$ about $s=0$, the above product is then

\begin{align} \label{(1-e^s) prod asymp}
    &\prod_{\substack{1\leq m\leq n\leq 2k\beta \\ \mu_m+\mu_n\neq0}}\Big(1-e^{-\frac{(v_m+v_n)}{N}-i(\mu_m+\mu_n)}\Big)^{-1} \nonumber \\
    &\qquad \sim\prod_{1\leq\sigma\leq\tau\leq k} \prod_{(m,n)\in\mathcal{U}_{\sigma,\tau;\underline{l}}^{+}} \frac{N}{v_m+v_n+i(t_{\sigma}+t_{\tau})} \prod_{(m,n)\in\mathcal{U}_{\sigma,\tau;\underline{l}}^{-}} \frac{N}{v_m+v_n-i(t_{\sigma}+t_{\tau})} \nonumber \\
    &\qquad\qquad \times\prod_{(m,n)\in\mathcal{V}_{\sigma,\tau;\underline{l}}^{+}} \frac{N}{v_m+v_n+i(t_{\sigma}-t_{\tau})} \prod_{(m,n)\in\mathcal{V}_{\sigma,\tau;\underline{l}}^{-}} \frac{N}{v_m+v_n-i(t_{\sigma}-t_{\tau})}.
\end{align}
The power of $N$ coming from this product is

\begin{equation}
    |\mathcal{T}_{k,\beta;\underline{l}}|=k\beta(2k\beta+1)-|\mathcal{A}_{k,\beta;\underline{l}}|.
\end{equation}

We therefore have that as $N\rightarrow\infty$, the integrals over $\theta_1,\dots,\theta_k$ are

\begin{align} \label{ref 3}
    & \int_0^{2\pi}\dotsi\int_0^{2\pi} \prod_{\substack{1\leq m\leq n\leq 2k\beta \\ \mu_m+\mu_n\neq0}}\left(1-e^{-\frac{(v_m+v_n)}{N}-i(\mu_m+\mu_n)}\right)^{-1} e^{iN \sum_{n=1}^{2k\beta}\mu_n} \prod_{m=1}^k d\theta_m \nonumber \\
    \sim& \int_0^{2N\pi}\dotsi\int_0^{2N\pi} \frac{N^{k\beta(2k\beta+1)-k-|\mathcal{A}_{k,\beta;\underline{l}}|} e^{2i\sum_{m=1}^k (l_m-\beta) t_m}}{\prod_{1\leq\sigma\leq\tau\leq k} \prod_{(m,n)\in\mathcal{U}_{\sigma,\tau;\underline{l}}^{+}} (v_m+v_n+i(t_{\sigma}+t_{\tau})) \prod_{(m,n)\in\mathcal{U}_{\sigma,\tau;\underline{l}}^{-}} (v_m+v_n-i(t_{\sigma}+t_{\tau}))} \nonumber \\
    &\qquad\qquad\qquad \times\frac{dt_1\cdots dt_k}{\prod_{(m,n)\in\mathcal{V}_{\sigma,\tau;\underline{l}}^{+}} (v_m+v_n+i(t_{\sigma}-t_{\tau})) \prod_{(m,n)\in\mathcal{V}_{\sigma,\tau;\underline{l}}^{-}} (v_m+v_n-i(t_{\sigma}-t_{\tau}))} \nonumber \\
    \sim&\, N^{k\beta(2k\beta+1)-k-|\mathcal{A}_{k,\beta;\underline{l}}|} \Psi_{k,\beta} (\underline{v};\underline{l}),
\end{align}
where

\begin{align} \label{Psi def}
    & \Psi_{k,\beta} (\underline{v};\underline{l}) \nonumber \\
    &\quad :=\int_0^{\infty}\cdots\int_0^{\infty} \frac{e^{2i\sum_{m=1}^k (l_m-\beta) t_m}}{\prod_{1\leq\sigma\leq\tau\leq k} \prod_{(m,n)\in\mathcal{U}_{\sigma,\tau;\underline{l}}^{+}} (v_m+v_n+i(t_{\sigma}+t_{\tau})) \prod_{(m,n)\in\mathcal{U}_{\sigma,\tau;\underline{l}}^{-}} (v_m+v_n-i(t_{\sigma}+t_{\tau}))} \nonumber \\
    &\qquad\qquad \times\frac{dt_1\cdots dt_k}{\prod_{(m,n)\in\mathcal{V}_{\sigma,\tau;\underline{l}}^{+}} (v_m+v_n+i(t_{\sigma}-t_{\tau})) \prod_{(m,n)\in\mathcal{V}_{\sigma,\tau;\underline{l}}^{-}} (v_m+v_n-i(t_{\sigma}-t_{\tau}))}.
\end{align}

Returning now to (\ref{MoM_Sp asymp}) and using (\ref{ref 3}), we have that

\begin{equation} \label{ref 4}
    \text{MoM}_{Sp(2N)} (k,\beta)\sim \gamma_{Sp} (k,\beta) N^{k\beta(2k\beta+1)-k},
\end{equation}
where

\begin{equation} \label{gamma_Sp def}
    \gamma_{Sp} (k,\beta):=\sum_{l_1,\dots,l_k=0}^{2\beta} \frac{(-1)^{k\beta+\sum_{m=1}^k l_m}c_{\underline{l}}(k,\beta)}{(2\pi)^k (2\pi i)^{2k\beta}(2k\beta)!} \int_{C_0}\cdots\int_{C_0} f(\underline{v};\underline{l}) \Psi_{k,\beta} (\underline{v};\underline{l}) \prod_{n=1}^{2k\beta} dv_n,
\end{equation}
and this completes the proof.

\end{proof}

\begin{proof} [Proof of theorem \ref{MoM Sp thm}]

To complete the proof of theorem \ref{MoM Sp thm}, we compare the asymptotic formula in lemma \ref{MoM_Sp asymp lemma} to that of theorem \ref{ABK Sp thm} to show that $\gamma_{Sp} (k,\beta)\neq 0$. We see that as $\text{MoM}_{Sp(2N)} (k,\beta)$ is a polynomial in $N$, we must have $\gamma_{Sp} (k,\beta)=\mathfrak{c}_{Sp} (k,\beta)>0$ which concludes the proof.

\end{proof}

\subsection{A symplectic family of $L$-functions}

For an example of a family of $L$-functions with symplectic symmetry, Bailey and Keating \cite{BK3} considered the family of quadratic Dirichlet $L$-functions. For $d$ a fundamental discriminant, let $\chi_d (n)=(\tfrac{d}{n})$ be the quadratic character defined by the Kronecker symbol. The associated $L$-function is defined for $\text{Re}(s)>1$ by

\begin{equation}
    L(s,\chi_d)=\sum_{n=1}^{\infty} \frac{\chi_d (n)}{n^s},
\end{equation}
and has an analytic continuation to $\mathbb{C}$. The $L$-function satisfies the functional equation

\begin{equation}
    L(s,\chi_d)=X_d (s) L(1-s,\chi_d),
\end{equation}
where $X_d (s)=|d|^{1/2-s} X(s,a)$ with $a=0$ if $d>0$ and $a=1$ if $d<0$, and

\begin{equation}
    X(s,a)=\pi^{s-\frac{1}{2}} \Gamma\left(\frac{1+a-s}{2}\right) \Gamma\left(\frac{s+a}{2}\right)^{-1}.
\end{equation}

The moments of moments of the family of quadratic Dirichlet $L$-functions are defined as

\begin{equation}
    \text{MoM}_{L_{\chi_d}} (k,\beta)=\frac{1}{D^*} \sideset{}{^*}\sum_{|d|\leq D} \left(\frac{1}{2\pi} \int_0^{2\pi}  L(\tfrac{1}{2}+i\theta,\chi_d)^{2\beta} d\theta\right)^k,
\end{equation}
where the sum is only over fundamental discriminants and $D^*$ is the number of terms in the sum. The conjecture of \cite{BK3} in this case is the following.

\begin{conjecture} [Bailey-Keating \cite{BK3}] \label{BK Sp conj}
For $k,\beta\in\mathbb{N}$, as $D\to\infty$,

\begin{equation}
    \textup{MoM}_{L_{\chi_d}} (k,\beta)=\eta(k,\beta) \mathfrak{c}_{Sp} (k,\beta) (\log D)^{k\beta(2k\beta+1)-k} \left(1+O_{k,\beta} (\log^{-1} D)\right),
\end{equation}
where $\mathfrak{c}_{Sp} (k,\beta)$ corresponds to the leading order coefficient in \textup{(\ref{ABK Sp result})} and $\eta(k,\beta)$ contains the arithmetic information.
\end{conjecture}

By adapting the proof of theorem \ref{MoM Sp thm}, we can relatively easily prove that conjecture \ref{BK Sp conj} follows from the shifted moment conjecture of \cite{CFKRS2}. Changing the order of integration and summation gives

\begin{equation}
    \text{MoM}_{L_{\chi_d}} (k,\beta)=\frac{1}{(2\pi)^k} \int_0^{2\pi}\dotsi\int_0^{2\pi} \frac{1}{D^*} \sideset{}{^*}\sum_{|d|\leq D} \prod_{m=1}^k L(\tfrac{1}{2}+i\theta_m,\chi_d)^{2\beta} d\theta_1\dots d\theta_k, 
\end{equation}
and the relevant shifted moment conjecture is the following.

\begin{conjecture} [Conrey et al. \cite{CFKRS2}]
Let $k,\beta\in\mathbb{N}$ and let $\underline{\theta}=(\theta_1,\dots,\theta_k)\in\mathbb{R}^k$. Then,

\begin{equation}
    \frac{1}{D^{*}} \sideset{}{^*}\sum_{|d|\leq D} \prod_{m=1}^k L(\tfrac{1}{2}+i\theta_m,\chi_d)^{2\beta}=\frac{1}{D^*} \sideset{}{^*}\sum_{|d|\leq D} \prod_{m=1}^k X_d (\tfrac{1}{2}+i\theta_m)^{\beta} Q_{k,\beta} (\log|d|,\underline{\theta})+O(D^{-\delta}),
\end{equation}
for some $\delta>0$, where

\begin{align} \label{Q(k,beta) def}
    Q_{k,\beta} (x,\underline{\theta})=& \frac{(-1)^{k\beta} 2^{2k\beta}}{(2\pi i)^{2k\beta} (2k\beta)!} \nonumber \\
    &\qquad \times\oint\dotsi\oint \frac{G(z_1,\dots,z_{2k\beta}) \Delta(z_1^2,\dots,z_{2k\beta}^2)^2 \prod_{n=1}^{2k\beta} z_n}{\prod_{n=1}^{2k\beta} \prod_{m=1}^k (z_n-i\theta_m)^{2\beta}(z_n+i\theta_m)^{2\beta}} e^{\frac{x}{2} \sum_{n=1}^{2k\beta}z_n} dz_1\dots dz_{2k\beta},
\end{align}
in which the path of integration encloses the poles at $\pm i\theta_m$ for $1\leq m\leq k$. Also,

\begin{equation}
    G(z_1,\dots,z_{2k\beta})=A_{k\beta} (z_1,\dots,z_{2k\beta}) \prod_{n=1}^{2k\beta} X(\tfrac{1}{2}+z_n,a)^{-\frac{1}{2}} \prod_{1\leq m\leq n\leq2k\beta} \zeta(1+z_m+z_n),
\end{equation}
where $A_{k\beta}$ is the Euler product, absolutely convergent for $|\textup{Re}(z_n)|<1/2$, defined by

\begin{align}
    A_{k\beta} (z_1,\dots,z_{2k\beta})=& \prod_p \prod_{1\leq m\leq n\leq 2k\beta} \left(1-\frac{1}{p^{1+z_m+z_n}}\right)\nonumber \\
    &\times\left(\frac{1}{2} \left(\prod_{n=1}^{2k\beta} \left(1-\frac{1}{p^{1/2+z_n}}\right)^{-1}+\prod_{n=1}^{2k\beta} \left(1+\frac{1}{p^{1/2+z_n}}\right)^{-1}\right)+\frac{1}{p}\right) \left(1+\frac{1}{p}\right)^{-1}.
\end{align}
\end{conjecture}

We therefore define

\begin{equation} \label{MoM_Q def}
    \text{MoM}_{Q_{k,\beta}} (D):=\frac{1}{(2\pi)^k} \int_0^{2\pi}\dotsi\int_0^{2\pi} \frac{1}{D^*} \sideset{}{^*}\sum_{|d|\leq D} \prod_{m=1}^k X_d (\tfrac{1}{2}+i\theta_m)^{\beta} Q_{k,\beta} (\log|d|,\underline{\theta})\, d\theta_1\dots d\theta_k,
\end{equation}
which should approximate $\text{MoM}_{L_{\chi_d}} (k,\beta)$ up to a power saving in $D$.

Comparing the integral $Q_{k,\beta} (x,\underline{\theta})$ with the integral expression for $I_{k\beta} (Sp(2N),\underline{\theta})$ in (\ref{I(Sp) def}), we immediately see the similarity on identifying $N$ with $x/2$. In particular, the product of $\zeta(1+z_m+z_n)$ terms replaces the product of $(1-e^{-z_m-z_n})^{-1}$ terms with both having the same analytic structure with simple poles at $z_m+z_n=0$. This means that the same analysis we applied to $I_{k\beta} (Sp(2N),\underline{\theta})$ can be applied to $Q_{k,\beta} (x,\underline{\theta})$ to yield an asymptotic formula for $\text{MoM}_{Q_{k,\beta}} (D)$. The function $G(z_1,\dots,z_{2k\beta})$ also contains arithmetic information in the Euler product $A_{k\beta}$ and the $X(s,a)$ factors. However, these factors do not present any additional difficulties. Using the fact that $A_{k\beta}$ is analytic in a neighbourhood of zero and $X(s,a)$ is analytic around $s=1/2$ and $X(\tfrac{1}{2},a)=1$, one can show that

    

\begin{align}
    \text{MoM}_{Q_{k,\beta}} (D)\sim & A_{k\beta} (0,\dots,0) \gamma_{Sp} (k,\beta) \frac{1}{D^{*}} \sum_{|d|\leq D} \left(\frac{\log |d|}{2}\right)^{k\beta(2k\beta+1)-k} \nonumber \\
    =& A_{k\beta} (0,\dots,0) \gamma_{Sp} (k,\beta) \left(\frac{\log D}{2}\right)^{k\beta(2k\beta+1)-k} \left(1+O(\log^{-1} D)\right),
\end{align}
where $\gamma_{Sp} (k,\beta)$ is as defined in (\ref{gamma_Sp def}). Thus, $\text{MoM}_{Q_{k,\beta}} (D)$ satisfies the asymptotic formula conjectured for $\text{MoM}_{L_{\chi_d}} (k,\beta)$ in conjecture \ref{BK Sp conj}.
 
\section{The special orthogonal group $SO(2N)$}

In this section we turn to the orthogonal case and prove theorem \ref{MoM SO thm}. From now on we assume that $k,\beta\in\mathbb{N}$ with $k,\beta$ not both 1. As in the symplectic case, the eigenvalues of matrices in $SO(2N)$ lie on the unit circle and come in complex conjugate pairs so

\begin{equation}
    \overline{P_{SO(2N)} (\theta;A)}=P_{SO(2N)} (-\theta;A).
\end{equation}
Then, as usual, we change the order of integration to write

\begin{equation}
    \text{MoM}_{SO(2N)} (k,\beta)=\frac{1}{(2\pi)^k} \int_0^{2\pi}\cdots\int_0^{2\pi} I_{k,\beta} (SO(2N),\theta_1,\dots,\theta_k)\, d\theta_1\cdots d\theta_k,
\end{equation}
where

\begin{equation}
    I_{k,\beta} (SO(2N),\underline{\theta}):=\int_{SO(2N)} \prod_{j=1}^k P_{SO(2N)} (\theta_j;A)^{\beta} P_{SO(2N)} (-\theta_j;A)^{\beta} dA.
\end{equation}
From \cite{CFKRS1}, we have the following contour integral expression for $I_{k,\beta}, (SO(2N),\underline{\theta})$ 

\begin{align} \label{I(SO) def}
    I_{k,\beta}(SO(2N),\underline{\theta})=& \frac{(-1)^{k\beta} 2^{2k\beta}}{(2\pi i)^{2k\beta} (2k\beta)!} \oint\dotsi\oint \prod_{1\leq m<n\leq2k\beta}(1-e^{-z_m-z_n})^{-1} \nonumber \\
    &\qquad \times\frac{\Delta(z_1^2,\dots,z_{2k\beta}^2)^2 \prod_{n=1}^{2k\beta} z_n}{\prod_{n=1}^{2k\beta} \prod_{m=1}^k (z_n-i\theta_m)^{2\beta} (z_n+i\theta_m)^{2\beta}} e^{N\sum_{n=1}^{2k\beta} z_n} dz_1\dots dz_{2k\beta},
\end{align}
where again the contours enclose the poles at $\pm i\theta_m$ for $1\leq m\leq k$. We note the similarity between the above expression for $I_{k,\beta} (SO(2N),\underline{\theta})$ and that for $I_{k,\beta} (Sp(2N),\underline{\theta})$ in (\ref{I(Sp) def}). Specifically, the only difference is in the product of $(1-e^{-z_m-z_n})^{-1}$ terms; in the symplectic case, the product is over $m\leq n$ rather than $m<n$. The proof of theorem \ref{MoM SO thm} will therefore mirror that of theorem \ref{MoM Sp thm} but with this slight difference.

First, by using lemma \ref{zero summands} and then following the proof of lemma \ref{I(Sp) asymp}, we get that

\begin{align}
    I_{k,\beta} (SO(2N),\underline{\theta})\sim& \sum_{l_1,\dots,l_k=0}^{2\beta} \frac{(-1)^{k\beta+\sum_{m=1}^k l_m} c_{\underline{l}} (k,\beta)}{(2\pi i)^{2k\beta} (2k\beta)!} N^{|\mathcal{A}_{k,\beta;\underline{l}}|} e^{iN \sum_{n=1}^{2k\beta} \mu_n} \nonumber \\
    &\qquad \times\int_{C_0}\cdots\int_{C_0} \prod_{\substack{1\leq m<n\leq 2k\beta \\ \mu_m+\mu_n\neq0}}\left(1-e^{-\frac{(v_m+v_n)}{N}-i(\mu_m+\mu_n)}\right)^{-1} f(\underline{v};\underline{l}) \prod_{n=1}^{2k\beta} dv_n,
\end{align}
where the $\mu_n$, the set $\mathcal{A}_{k,\beta;\underline{l}}$ and the function $f(\underline{v};\underline{l})$ are as defined in (\ref{mu def}), (\ref{set A def}) and (\ref{f(v) def}) respectively. We then proceed as in the proof of lemma \ref{MoM_Sp asymp lemma} with the change being that we will replace the set $\mathcal{T}_{k,\beta;\underline{l}}$ by

\begin{equation}
    \overset{\sim}{\mathcal{T}}_{k,\beta;\underline{l}}:=\{(m,n):1\leq m<n\leq 2k\beta, \mu_m+\mu_n\neq0\}=\{(m,n):1\leq m<n\leq 2k\beta\}\setminus\mathcal{A}_{k,\beta;\underline{l}},
\end{equation}
and for $1\leq\sigma\leq\tau\leq k$, define the subsets

\begin{equation}
    \overset{\sim}{\mathcal{U}}_{\sigma,\tau;\underline{l}}^{+}:=\{(m,n)\in\overset{\sim}{\mathcal{T}}_{k,\beta;\underline{l}}: \mu_m+\mu_n=\theta_{\sigma}+\theta_{\tau}\}
\end{equation}

\begin{equation}
    \overset{\sim}{\mathcal{U}}_{\sigma,\tau;\underline{l}}^{-}:=\{(m,n)\in\overset{\sim}{\mathcal{T}}_{k,\beta;\underline{l}}: \mu_m+\mu_n=-(\theta_{\sigma}+\theta_{\tau})\},
\end{equation}
and

\begin{equation}
    \overset{\sim}{\mathcal{V}}_{\sigma,\tau;\underline{l}}^{+}:=\{(m,n)\in\overset{\sim}{\mathcal{T}}_{k,\beta;\underline{l}}: \mu_m+\mu_n=\theta_{\sigma}-\theta_{\tau}\}
\end{equation}

\begin{equation}
    \overset{\sim}{\mathcal{V}}_{\sigma,\tau;\underline{l}}^{-}:=\{(m,n)\in\overset{\sim}{\mathcal{T}}_{k,\beta;\underline{l}}: \mu_m+\mu_n=-(\theta_{\sigma}-\theta_{\tau})\}.
\end{equation}
After making the same change of variables $t_m=N\theta_m$, the product of $(1-e^{-z_m-z_n})^{-1}$ terms will be

\begin{align} \label{SO (1-e^s) prod asymp}
    &\prod_{\substack{1\leq m<n\leq 2k\beta \\ \mu_m+\mu_n\neq0}} \Big(1-e^{-\frac{(v_m+v_n)}{N}-i(\mu_m+\mu_n)}\Big)^{-1} \nonumber \\
    &\qquad \sim\prod_{1\leq\sigma\leq\tau\leq k} \prod_{(m,n)\in\overset{\sim}{\mathcal{U}}_{\sigma,\tau;\underline{l}}^{+}} \frac{N}{v_m+v_n+i(t_{\sigma}+t_{\tau})} \prod_{(m,n)\in\overset{\sim}{\mathcal{U}}_{\sigma,\tau;\underline{l}}^{-}} \frac{N}{v_m+v_n-i(t_{\sigma}+t_{\tau})} \nonumber \\
    &\qquad\qquad \times\prod_{(m,n)\in\overset{\sim}{\mathcal{V}}_{\sigma,\tau;\underline{l}}^{+}} \frac{N}{v_m+v_n+i(t_{\sigma}-t_{\tau})} \prod_{(m,n)\in\overset{\sim}{\mathcal{V}}_{\sigma,\tau;\underline{l}}^{-}} \frac{N}{v_m+v_n-i(t_{\sigma}-t_{\tau})}.
\end{align}
The power of $N$ coming from this product is

\begin{equation}
    |\overset{\sim}{\mathcal{T}}_{k,\beta;\underline{l}}|=k\beta (2k\beta-1)-|\mathcal{A}_{k,\beta;\underline{l}}|.
\end{equation}

Taking into account this difference, we see that in this case, we will obtain

\begin{equation} \label{ref 5}
    \text{MoM}_{SO(2N)} (k,\beta)\sim\gamma_{SO} (k,\beta) N^{k\beta(2k\beta-1)-k},
\end{equation}
where

\begin{equation} \label{gamma_SO def}
    \gamma_{SO} (k,\beta):=\sum_{l_1,\dots,l_k=0}^{2\beta} \frac{(-1)^{k\beta+\sum_{m=1}^k l_m}c_{\underline{l}}(k,\beta)}{(2\pi)^k (2\pi i)^{2k\beta}(2k\beta)!} \int_{C_0}\cdots\int_{C_0} f(\underline{v};\underline{l}) \Omega_{k,\beta} (\underline{v};\underline{l}) \prod_{n=1}^{2k\beta} dv_n, 
\end{equation}
and

\begin{align} \label{Omega def}
    & \Omega_{k,\beta} (\underline{v};\underline{l}) \nonumber \\
    &\quad :=\int_0^{\infty}\cdots\int_0^{\infty} \frac{e^{2i\sum_{m=1}^k (l_m-\beta) t_m}}{\prod_{1\leq\sigma\leq\tau\leq k} \prod_{(m,n)\in\overset{\sim}{\mathcal{U}}_{\sigma,\tau;\underline{l}}^{+}} (v_m+v_n+i(t_{\sigma}+t_{\tau})) \prod_{(m,n)\in\overset{\sim}{\mathcal{U}}_{\sigma,\tau;\underline{l}}^{-}} (v_m+v_n-i(t_{\sigma}+t_{\tau}))} \nonumber \\
    &\qquad\qquad \times\frac{dt_1\cdots dt_k}{\prod_{(m,n)\in\overset{\sim}{\mathcal{V}}_{\sigma,\tau;\underline{l}}^{+}} (v_m+v_n+i(t_{\sigma}-t_{\tau})) \prod_{(m,n)\in\overset{\sim}{\mathcal{V}}_{\sigma,\tau;\underline{l}}^{-}} (v_m+v_n-i(t_{\sigma}-t_{\tau}))}.
\end{align}

Finally, comparing the asymptotic formula (\ref{ref 5}) to the result of theorem \ref{ABK SO thm} shows that \\ $\gamma_{SO} (k,\beta)=\mathfrak{c}_{SO} (k,\beta)>0$ which completes the proof of theorem \ref{MoM SO thm}.

\subsection{An orthogonal family of $L$-functions}

An example of a family of $L$-functions with orthogonal symmetry is the family of quadratic twists of an elliptic curve $L$-function. Let $E$ be an elliptic curve defined over $\mathbb{Q}$ with conductor $M$. The $L$-function attached to $E$ is defined for $\text{Re}(s)>1$ by

\begin{equation}
    L_E(s)=\sum_{n=1}^{\infty} \frac{a_n}{n^{s+1/2}}=\prod_{p|M} (1-a_p p^{-s-\frac{1}{2}})^{-1} \prod_{p\nmid M} (1-a_p p^{-s-\frac{1}{2}}+p^{-2s})^{-1}:=\prod_{p}\mathcal{L}_p(p^{-s}),
\end{equation}
where the $a_p$ are related to the number of points on the reduction of $E$ mod $p$. $L_E(s)$ can be analytically continued to $\mathbb{C}$ and satisfies the functional equation

\begin{equation}
    L_E(s)=w_E Y(s) L_E(1-s),
\end{equation}
where $w_E=\pm1$ is the sign of the functional equation and

\begin{equation}
    Y(s)=\left(\frac{\sqrt{M}}{2\pi}\right)^{1-2s} \Gamma\left(\frac{3}{2}-s\right) \Gamma\left(\frac{1}{2}+s\right)^{-1}.
\end{equation}
For $d$ a fundamental discriminant with $(d,M)=1$, the twist of $L_E(s)$ by the quadratic character $\chi_d(n)=(\tfrac{d}{n})$ is

\begin{equation}
    L_E(s,\chi_d)=\sum_{n=1}^{\infty} \frac{a_n \chi_d(n)}{n^{s+1/2}}.
\end{equation}
These twisted $L$-functions can also be analytically continued to $\mathbb{C}$ and they satisfy the functional equation

\begin{equation}\label{L_E functional eq}
    L_E(s,\chi_d)=w_E \chi_d(-M) Y_d(s) L_E(1-s,\chi_d),
\end{equation}
where $Y_d(s)=|d|^{1-2s}Y(s)$. The set of $L_E(s,\chi_d)$ for which the sign $w_E \chi_d(-M)$ of the functional equation equals $+1$ forms a family with even orthogonal symmetry and so we use the special orthogonal group $SO(2N)$ for comparison.

The moments of moments of this family are defined as

\begin{equation}
    \text{MoM}_{L_E} (k,\beta)=\frac{1}{D^*} \sideset{}{^*}\sum_{\substack{|d|\leq D\\ w_E \chi_d(-M)=1}} \left(\frac{1}{2\pi} \int_0^{2\pi}  L_E (\tfrac{1}{2}+i\theta,\chi_d)^{2\beta} d\theta\right)^k,
\end{equation}
where the sum is only over fundamental discriminants and $D^{*}$ is the number of terms in the sum. The conjecture made in \cite{BK3} for this family, based on theorem \ref{ABK SO thm}, is

\begin{conjecture} [Bailey-Keating \cite{BK3}] \label{BK SO conj}
    For $k,\beta\in\mathbb{N}$ and $k,\beta$ not both 1, as $D\rightarrow\infty$,
    
    \begin{equation}
        \textup{MoM}_{L_E}(k,\beta)=\xi(k,\beta) \mathfrak{c}_{SO}(k,\beta) (\log D)^{k\beta(2k\beta-1)-k} \big(1+O_{k,\beta}\big(\log^{-1}D\big)\big),
    \end{equation}
    where $\mathfrak{c}_{SO}(k,\beta)$ corresponds to the leading order coefficient in \textup{(\ref{ABK SO result})} and $\xi(k,\beta)$ contains the arithmetic information. 
\end{conjecture}

Once again, we can change the order of integration and summation to write

\begin{equation}
    \text{MoM}_{L_E} (k,\beta)=\frac{1}{(2\pi)^k} \int_0^{2\pi}\cdots\int_0^{2\pi} \frac{1}{D^*} \sideset{}{^*}\sum_{\substack{|d|\leq D\\ w_E \chi_d(-M)=1}} \prod_{m=1}^k L_E (\tfrac{1}{2}+i\theta_m,\chi_d)^{2\beta} d\theta_1\cdots d\theta_k,
\end{equation}
and we have the following conjecture of \cite{CFKRS2} for the shifted moments in the integrand.

\begin{conjecture} [Conrey et al. \cite{CFKRS2}]
    Let $k,\beta\in\mathbb{N}$ and let $\underline{\theta}=(\theta_1,\dots,\theta_k)\in\mathbb{R}^k$. Then,
    
    \begin{equation}
        \frac{1}{D^*} \sideset{}{^*}\sum_{\substack{|d|\leq D\\ w_E \chi_d(-M)=1}} \prod_{m=1}^k L_E (\tfrac{1}{2}+i\theta_m,\chi_d)^{2\beta}=\frac{1}{D^*} \sideset{}{^*}\sum_{\substack{|d|\leq D\\ w_E \chi_d(-M)=1}} \prod_{m=1}^k Y_d (\tfrac{1}{2}+i\theta_m)^{\beta} \Upsilon_{k,\beta} (\log |d|,\underline{\theta})+O(D^{-\delta}),
    \end{equation}
    for some $\delta>0$, where
    
    \begin{align} \label{Upsilon(k,beta) def}
        \Upsilon_{k,\beta} (x,\underline{\theta})=& \frac{(-1)^{k\beta} 2^{2k\beta}}{(2\pi i)^{2k\beta} (2k\beta)!} \nonumber \\
        &\qquad \times\oint\dotsi\oint \frac{H(z_1,\dots,z_{2k\beta}) \Delta(z_1^2,\dots,z_{2k\beta}^2)^2 \prod_{n=1}^{2k\beta} z_n}{\prod_{n=1}^{2k\beta} \prod_{m=1}^k (z_n-i\theta_m)^{2\beta}(z_n+i\theta_m)^{2\beta}} e^{x \sum_{n=1}^{2k\beta}z_n} dz_1\dots dz_{2k\beta},
    \end{align}
    in which the path of integration encloses the poles at $\pm i\theta_m$ for $1\leq m\leq k$. Also,
    
    \begin{equation}
        H(z_1,\dots,z_{2k\beta})=B_{k\beta} (z_1,\dots,z_{2k\beta}) \prod_{n=1}^{2k\beta} Y(\tfrac{1}{2}+z_n)^{-\frac{1}{2}} \prod_{1\leq m<n\leq 2k\beta} \zeta(1+z_m+z_n),
    \end{equation}
    where $B_{k\beta}$ is the Euler product, absolutely convergent for $|\textup{Re}(z_n)|<1/2$, defined by
    
    \begin{align}
        B_{k\beta} (z_1,\dots,z_{2k\beta})=& \prod_p \prod_{1\leq m<n\leq 2k\beta} \left(1-\frac{1}{p^{1+z_m+z_n}}\right)\nonumber \\
        &\times\left(\frac{1}{2} \left(\prod_{n=1}^{2k\beta} \mathcal{L}_p \left(\frac{1}{p^{1/2+z_n}}\right)+\prod_{n=1}^{2k\beta} \mathcal{L}_p \left(\frac{-1}{p^{1/2+z_n}}\right)\right)+\frac{1}{p}\right) \left(1+\frac{1}{p}\right)^{-1}.
\end{align}
\end{conjecture}

Naturally, we define

\begin{equation}
    \text{MoM}_{\Upsilon_{k,\beta}} (D):=\frac{1}{(2\pi)^k} \int_0^{2\pi}\cdots\int_0^{2\pi} \frac{1}{D^*} \sideset{}{^*}\sum_{\substack{|d|\leq D\\ w_E \chi_d(-M)=1}}  \prod_{m=1}^k Y_d (\tfrac{1}{2}+i\theta_m)^{\beta} \Upsilon_{k,\beta} (\log |d|,\underline{\theta})\, d\theta_1\cdots d\theta_k,
\end{equation}
which should approximate $\text{MoM}_{L_E} (k,\beta)$ up to a power saving in $D$. Similarly to the symplectic case considered earlier, we can clearly see the similarity between the integral expressions for $\Upsilon_{k,\beta} (x,\underline{\theta})$ above and $I_{k,\beta} (SO(2N),\underline{\theta})$ in (\ref{I(SO) def}). Hence, by following the proof of theorem \ref{MoM SO thm} and taking into account the arithmetic factors just as in the case of the quadratic Dirichlet $L$-functions, one can show that

\begin{align}
    \text{MoM}_{\Upsilon_{k,\beta}} (D)\sim& B_{k\beta} (0,\dots,0) \gamma_{SO} (k,\beta) \frac{1}{D^*} \sideset{}{^*}\sum_{\substack{|d|\leq D\\ w_E \chi_d(-M)=1}} (\log |d|)^{k\beta(2k\beta-1)-k} \nonumber \\
    =& B_{k\beta} (0,\dots,0) \gamma_{SO} (k,\beta) (\log D)^{k\beta(2k\beta-1)-k} \left(1+O(\log^{-1} D)\right),
\end{align}
where $\gamma_{SO} (k,\beta)$ is as defined in (\ref{gamma_SO def}). Therefore, conjecture \ref{BK SO conj} also follows from the shifted moment conjecture of \cite{CFKRS2}.

\vspace{0.5cm}

\noindent \textit{Acknowledgment.}  
The first author is grateful to the Leverhulme Trust (RPG-2017-320) for the support through the research project grant ``Moments of $L$-functions in Function Fields and Random Matrix Theory". The research of the second author is supported by an EPSRC Standard Research Studentship (DTP) at the University of Exeter.

\end{document}